\documentclass[journal]{IEEEtran}

\usepackage{amsmath,amssymb,amsfonts,amsthm}
\usepackage{xfrac}
\usepackage{mathtools}
\usepackage{float}
\usepackage{caption}
\usepackage{subcaption}
\usepackage[capitalize]{cleveref} % cool stuff for cross-reference
\usepackage{makecell}
\usepackage{enumitem}  
\usepackage{algorithm}
\usepackage{algpseudocode}

\usepackage{pgfplots}
\usepackage{tikz-network}
\DeclareUnicodeCharacter{2212}{−}
\usepgfplotslibrary{groupplots,dateplot}
\usetikzlibrary{patterns,shapes.arrows}
\pgfplotsset{compat=newest}

\usepackage{pifont}

\DeclareMathAlphabet{\mymathbb}{U}{BOONDOX-ds}{m}{n}

\newcounter{thms}
\newtheorem{theorem}[thms]{Theorem}
\newtheorem{corollary}[thms]{Corollary}
\newtheorem{lemma}[thms]{Lemma}

\newcounter{defs}
\newtheorem{definition}[defs]{Definition}

\theoremstyle{remark}

\newcommand{\symi}{i}
\newcommand{\symk}{k}
\newcommand{\syml}{\ell}
\newcommand{\symt}{t}

\newcommand{\R}{\mathbb{R}}

\newcommand{\bb}[1]{\mymathbb{#1}}

\newcommand{\mc}[1]{\mathcal{#1}}
\newcommand{\mt}[1]{\textrm{#1}}

\newcommand{\lt}{\left}
\newcommand{\rt}{\right}
\newcommand{\beeq}{\begin{equation}}
\newcommand{\eneq}{\end{equation}}
\newcommand{\matb}{\begin{matrix}}
\newcommand{\mate}{\end{matrix}}

\newcommand{\unaryminus}{\scalebox{0.6}[1.0]{$ - $}}

\DeclareMathOperator*{\argmin}{arg\,min}

\DeclareMathOperator*{\vect}{\mt{vec}}

\title{\LARGE \bf
Robust online joint state/input/parameter estimation of linear systems
}

\author{Jean-Sébastien Brouillon, Keith Moffat, Florian Dörfler, and Giancarlo Ferrari-Trecate% <-this % stops a space
\thanks{This research is supported by the Swiss National Science Foundation under the NCCR Automation (grant agreement 51NF40\_180545).}% <-this % stops a space
\thanks{J.S. Brouillon and G. Ferrari-Trecate are with the Institute of Mechanical Engineering, \'Ecole Polytechnique F\'ed\'erale de Lausanne, Switzerland. Email addresses: {\tt\small \{jean-sebastien.brouillon, giancarlo.ferraritrecate\}@epfl.ch}.}%
\thanks{K. Moffat is with the Electrical Engineering and Computer Science Department, UC Berkeley, USA. Email address: {\tt\small keithm@berkeley.edu}.}%
\thanks{F. Dörfler is with the Automatic Control Laboratory, Swiss Federal Institute of Technology (ETH), Switzerland. Email address: {\tt\small dorfler@control.ee.ethz.ch}.}%
}

\begin{document}
\bstctlcite{IEEEexample:BSTcontrol}

\maketitle
\thispagestyle{empty}
\pagestyle{empty}

%%%%%%%%%%%%%%%%%%%%%%%%%%%%%%%%%%%%%%%%%%%%%%%%%%%%%%%%%%%%%%%%%%%%%%%%%%%%%%%%
\begin{abstract}
 This paper presents a method for jointly estimating the state, input, and parameters of linear systems in an online fashion. The method is specially designed for measurements that are corrupted with non-Gaussian noise or outliers, which are commonly found in engineering applications. In particular, it combines recursive, alternating, and iteratively-reweighted least squares into a single, one-step algorithm, which solves the estimation problem online and benefits from the robustness of least-deviation regression methods. The convergence of the iterative method is formally guaranteed. Numerical experiments show the good performance of the estimation algorithm in presence of outliers and in comparison to state-of-the-art methods.

\end{abstract}

%%%%%%%%%%%%%%%%%%%%%%%%%%%%%%%%%%%%%%%%%%%%%%%%%%%%%%%%%%%%%%%%%%%%%%%%%%%%%%%%
\section{Introduction} \label{section_intro}

Reliable control of a dynamical system is often contingent upon an accurate system model and accurate state measurements. When the system dynamics change over time, system identification algorithms can be used to track the parameter variations in the system by using state and input samples. However, measurement noise and inaccurate actuation may strongly degrade the estimation accuracy. In this case, joint state/input and parameter estimation methods are required in order to estimate both the noiseless states/inputs and the unknown parameters. In addition, in order to be useful for control applications, the estimation algorithm must provide all estimates in real time. Online joint state/input and parameter estimation methods are used in power systems, battery management systems, and self-driving cars, among other applications \cite[Section 3.3.3]{power_systems_joint}, \cite{ekf_tuning_beelen2020, ukf_wielitzka2015}. While the constantly growing range of sensing technologies often allows one to measure the complete state of a system, statistical outliers and non-Gaussian noise, introduced by complex sensors, faulty hardware, and cyberattacks, are often present \cite{nongauss_vision_ieng2007, cyber_attack_outliers}.

In these scenarios, very popular estimation algorithms built for input/output (rather than state/input) samples and assuming Gaussian noise distributions can be ineffective. The Extended Kalman Filter (EKF), which has been demonstrated to work well for moderate and Gaussian measurement noise, requires additional constraints and tuning to remain stable when the noise levels increase \cite{ekf_kullberg2021online, ekf_joint_stateparaminput_naets2015, joint_efk_plett2005, ekf_stab_rapp2004, limits_ekf_castellanos2004, ekf_tuning_beelen2020}. Because the joint estimation problem for linear systems is bilinear, the Unscented Kalman Filter (UKF) can achieve higher accuracy as it relies on sigma points approximating Gaussian distributions rather than noise-sensitive linearizations \cite{ukf_wielitzka2015}. However, the sigma-point approximations can still introduce significant inaccuracies for non-Gaussian noise. Recursive Alternating Least Squares and Dual Kalman Filters provide an alternative way to deal with the bilinearity by considering state/input and parameters separately \cite{rals_khouaja2004, dkf_zwerger2020}. Error-in-variables (EIV) approaches, which include Total Least Squares (TLS) methods, can also be used for joint state and parameter estimation of linear systems, as EIV estimators implicitly reconstruct the inputs/states in order to find the parameters that best fit the data \cite{huffel_tls, rhode2014_rgtls}. Among EIV methods, Recursive Total Least Squares (RTLS) has been shown to outperform Kalman filtering for linear systems \cite{rhode2014_rgtls}. Similarly to EIV models, subspace identification uses noisy measurement to provide parameter estimates, which can then be used to find the most likely states and inputs of the system \cite{subspace_van2012, recursive_subspace_mercer}.

The impact of outliers can be minimized using outlier detection techniques to pre-filter the data \cite{outlier_detection_book}. However, they usually amount to solve complex classification problems, which are prohibitive to solve in real time. %This paper focuses on joint estimation that is \emph{robust} to outliers and non-standard noise distributions.
A standard approach to improve robustness against outliers, and that can be adapted to online methods, is to consider tail-heavy noise distributions and/or to add regularization terms such as $\ell_1$ norm penalties \cite{least_deviations_book, bertsimas2018_robust_opt, bar_def, robust_subspace_sadigh2014, robust_ekf_longhini2021, distrirobust_shafieezadeh2019}.% These tail-heavy/regularization approaches are better suited for online implementation.

The non-smoothness of the $\ell_1$ norm creates complicated optimization problems when applied to the standard methods mentioned previously.
ADMM, proximal, and sub-gradient methods have been used to solve optimization problems with $\ell_1$ costs \cite{online_admm, proximal_methods}. However, such approaches are tailored to batch estimation problems. Another method for dealing with $\ell_1$ terms is iterative reweighting, which consists of iteratively approximating a class of cost functions by quadratic ones, and by only relying on the current value of the cost \cite{irls_convergence}. This method is simpler, compared to the aforementioned alternatives, because it does not require a gradient or proximal operator computation.
While convergence proofs exist for Iteratively-Reweighted Least Squares and its variants, they are not valid in a online (i.e. recursive) or alternating setting \cite{irls_convergence, bar_asympto}. 

% This paper provides a state, input and parameter estimation algorithm that addresses the three challenges of online, joint and robust estimation. To do so, we combine the benefits of three state-of-the-art solutions for each problem, respectively: Recursive Least Squares, Alternating Least Squares, and Iteratively-Reweighed Least Squares.  
% First, the recursive nature of the algorithm allows it to run online.
% Second, the alternating portion allows the algorithm to simultaneously estimate both the parameters and the states/inputs.
% Finaly, iterative reweighting allows it to approximate $\ell_1$ terms in the cost function, making the resulting estimate robust to outliers. 

In this paper, we propose a new Alternating and Iteratively-reweighed Recursive Least Squares algorithm (AIRLS) that addresses the three challenges of online, joint and robust estimation. A second contribution is the formal proof of convergence of the method. Finally, we present numerical experiments demonstrating the robustness of AIRLS to outliers. More in details, we show that, in presence of outliers, both EKF and a standard subspace identification method \cite{ref_n4sid} can fail, even for a very simple system. In the same setting, RTLS shows greater robustness, but is outperformed by AIRLS.
%The computationally-efficient AIRLS algorithm performs well, despite the outlier noise, accurately estimating the system parameters, states, and inputs.

The paper is structured as follows: \cref{section_prob} describes the joint estimation problem to be solved. A naive form of robust and online joint estimation is introduced in \cref{section_airls} and then used in \cref{section_conv} to define the AIRLS method and prove its convergence. \cref{section_results} presents the numerical results.

\subsection{Preliminaries and Notation}
\label{subsec:notation}
The $n$-dimensional identity matrix is $I_n \in \R^{n \times n}$, and $\bb 0_n \in \R^n$ and $\bb 1_n \in \R^n$ are the vectors of all zeros and all ones. $[x_i^\top]_{\symi=1}^n$ is the matrix with rows equal to the vectors $x_\symi$. $\vect(\cdot)$ is the column vectorization operator. A proportionality relation is denoted by $\propto$. The $\ell_1$ Frobenius norm is defined as  $\|A\|_{F1} = \|\vect(A)\|_1 = \sum_{\symi} \|[A]_{\symi}\|_1$, where $[A]_{\symi}$ is the $i^{th}$ column of $A$. Similarly, the $\ell_2$ Frobenius norm is $\|A\|_{F} = \|\vect(A)\|_2 = \sqrt{\sum_{\symi} \|[A]_{\symi}\|_2^2}$.

\section{Problem Statement}\label{section_prob}

\subsection{System model}\label{subsec_system_def}

We want to reconstruct the matrices $A$ and $B$ in the model
\begin{align}\label{eq_noisless_model}
    x_{\symt+1} = A x_{\symt} + B u_{\symt},
\end{align}
where $u_\symt \in \R^{n_u}$ are the inputs and $x_\symt \in \R^{n}$ are the states for the time steps $t=1,\dots, N$.
We also want to provide a running estimate of $x_\symt$ and $u_\symt$ from measurements $\tilde x_\symt$ and $\tilde u_\symt$ that are corrupted by additive noise $\Delta x_\symt$ and $\Delta u_\symt$, i.e.
\begin{subequations}
\begin{align}
    \tilde x_{\symt} = x_{\symt} + \Delta x_{\symt},\\
    \tilde u_{\symt} = u_{\symt} + \Delta u_{\symt}.
\end{align}
\end{subequations}
Note that we do not assume any specific probability distribution for the noise.

%\section{Background}\label{section_back}

%\subsection{Correlation matrix}\label{subsec_batch_prob}

%The observation at time $\symt+1$ consists of $n$ state measurements at time $\symt+1$, $n$ state measurements at time $\symt$, and $n_u$ input measurements at time $\symt$.
A recursive algorithm for solving this problem must be based on a fixed-size data matrix and a constant number of parameters. To this purpose, a common approach is to replace the observations $[x_{\symt+1}, x_\symt, u_\symt]$ by their discounted empirical correlation matrix
\begin{align}\label{eq_autocorr_def}
    %C_\symt =\! \lt[\!\beta^{\symt \unaryminus 1}\!\! \lt[\matb {x}_1 \\ {x}_0 \\  u_{0} \mate\rt]\!\!, \dots,\!\! \lt[\matb  x_{\symt+1} \\  x_{\symt} \\  u_{\symt} \mate\rt]\!\rt]\!\! \lt[\! \lt[\matb {x}_1 \\ {x}_0 \\  u_{0} \mate\rt]\!\!, \dots,\!\! \lt[\matb  x_{\symt+1} \\  x_{\symt} \\  u_{\symt} \mate\rt]\!\rt]^{\!\!\!\top}\!\!.
    C_\symt = \sum_{\symi = 0}^\symt \beta^{\symt-\symi} \Gamma_\symi \quad: \quad \Gamma_\symi = \lt[\matb  x_{\symi+1} \\  x_{\symi} \\  u_{\symi} \mate\rt] \lt[\matb  x_{\symi+1} \\  x_{\symi} \\  u_{\symi} \mate\rt]^\top,
\end{align}
%The matrix $C_\symt \in \R^{m \times m}$, $m=2n+n_u$ is defined as a sum of quadratic forms and \eqref{eq_noisless_model} constrains $x_{\symt+1}$ to be a linear combination of $x_{\symt}$ and $u_\symt$ for all $\symt$. Therefore, under persistent excitation, $C_\symt$ is positive semi-definite of rank $n+n_u$. 
where $1-\beta \in (0, 1)$ is the forgetting factor. In presence of noise, one can only build $\tilde C_\symt$ and $\tilde \Gamma_\symt$ from measurements $\tilde x_0, \dots, \tilde x_{\symt+1}$ and $\tilde u_0, \dots, \tilde u_{\symt}$. We note that, usually, the matrix $\tilde C_\symt$ is full rank due to the noise.

%\subsection{Standard methods}
Next, we briefly review standard online estimation methods. The RTLS algorithms estimates the null space of $C_\symt$ by identifying the eigenvectors of $\tilde C_\symt$ corresponding to its smallest eigenvalues. This can be done with the inverse power method \cite{rhode2014_rgtls}. The null space is formed by vectors $[x_{\symt+1}^\top, x_\symt^\top, u_\symt^\top]^\top$ verifying \eqref{eq_noisless_model} and therefore allows one to retrieve the matrices $A$ and $B$.

Recursive subspace identification and the EKF recursively estimate propagator matrices, which have an expression similar to $C_\symt$, and use their inverse for updating the estimates of $A$ and $B$ \cite{recursive_subspace_mercer, ekf_kullberg2021online}.

\section{Online robust joint state/input and parameter estimation}\label{section_airls}

In this section, we provide a preliminary online estimation algorithm using $\tilde{C}_\symt$ for estimating both $[A,B]$ and $C_\symt$. Moreover, in \cref{section_conv}, we show that the estimate of states and inputs at time $\symt$ can be recovered by projecting the noisy measurements $\tilde x_\symt$ and $\tilde u_\symt$ onto the null space of $[-I_n, \hat A_\symt, \hat B_\symt]$ (see \cref{cor_proj} below).

%\subsection{Partitioned and corrected correlation}\label{subsec_method_intro}

%In order to ensure a rank $n+n_u$ for $C_\symt$, 
Before presenting the algorithm in Sections \ref{subsec_alt_method} and \ref{subsec_algo1}, we need to introduce some notations. We partition $C_\symt = [Y_\symt^\top, Z_\symt^\top]^\top$, where $Y_\symt \in \R^{n \times m}$ and $Z_\symt \in \R^{(n+n_u) \times m}$ are two linearly dependent blocks. Using \eqref{eq_noisless_model} and \eqref{eq_autocorr_def}, this gives
\begin{align}\label{eq_noiseless_model_recursive}
    Y_\symt = E_y C_\symt = [A, B] E_z C_\symt = [A,B] Z_\symt,
\end{align}
where $E_y = [I_n, \bb 0_{n \times (n+n_u)}]$ and $E_z = [\bb 0_{(n+n_u) \times n}, I_{n+n_u}]$. Similar to $C_\symt$, the matrix $\tilde C_\symt$ is partitioned into $\tilde Y_\symt = E_y \tilde C_\symt$ and $\tilde Z_\symt = E_z \tilde C_\symt$.

Because $Y_\symt$ depends on $A$, $B$, and $Z_\symt$, the estimation problem amounts to finding the estimates $\hat A_\symt$, $\hat B_\symt$, and $\hat Z_\symt$ from $\tilde C_\symt$. To improve accuracy, one can use the previous estimate $\hat Z_{\symt-1}$, rather than $\tilde Z_{\symt-1}$. This amounts to using $\underline C_\symt$ rather than $\tilde C_\symt$, which we define recursively as
%This can be done by rewriting \eqref{eq_autocorr_def} for $\underline C_\symt$ in a recursive way:
\begin{align}\label{eq_autocorr_rec}
    \underline C_{\symt+1} = \beta \underline C_{\symt} + \tilde \Gamma_{\symt+1}=  \beta \lt[\matb \tilde Y_{\symt} \\ \hat Z_{\symt} \mate\rt] + \tilde \Gamma_{\symt+1} ,
\end{align}
% We define the corrected correlation matrix $\underline C_\symt$ by replacing $\tilde Z_{\symt-1}$ by $\hat Z_{\symt-1}$ in \eqref{eq_autocorr_rec}. 
where $\underline C_0 \propto I_m$ is arbitrarily chosen. Similar to $\tilde C_\symt$, the matrix $\underline C_\symt$ is partitioned into $\underline Y_\symt = E_y \underline C_\symt = \tilde Y_\symt$ and $\underline Z_\symt = E_z \underline C_\symt$.

% In order to ensure a rank $n+n_u$ for $C_\symt$, we will partition it into two blocks, $Y_\symt \in \R^{n \times m}$ and $Z_\symt \in \R^{(n+n_u) \times m}$. This gives $C_\symt = [Y_\symt^\top, Z_\symt^\top]^\top$.

% \begin{lemma}\label{lem_rec}
% If $x_\symt$ and $u_\symt$ follow \eqref{eq_noisless_model} for all $\symt$, then
% \begin{align}\label{eq_noiseless_model_recursive}
%     Y_\symt = E_y C_\symt = [A, B] E_z C_\symt = [A,B] Z_\symt,
% \end{align}
% where $E_y = [I_n, \bb 0_{n \times (n+n_u)}]$ and $E_z = [\bb 0_{(n+n_u) \times n}, I_{n+n_u}]$.
% \end{lemma}
% \begin{proof}
% We start by rewriting \eqref{eq_noisless_model} for all $\symt$:
% \begin{align}\label{eq_proof_noisless_stack}
%     [x_1, \dots, x_{\symt+1}] = [A,B] \lt[\matb x_0 & \dots & x_{\symt} \\ u_0 & \dots & u_{\symt} \mate\rt],
% \end{align}
% which is equivalent to
% \begin{align}\label{eq_proof_noisless_stack_discount}
%     [\beta^{\symt \unaryminus 1} x_1, \dots, x_{\symt+1}] = [A,B] \!\lt[\beta^{\symt \unaryminus 1}\! \lt[\matb {x}_0 \\  u_{0} \mate\rt]\!, \dots,\! \lt[\matb  x_{\symt} \\  u_{\symt} \mate\rt]\rt]\!.
% \end{align}
% If we post-multiply both sides of the equality by $\lt[\lt[\matb {x}_1 \\ {x}_0 \\  u_{0} \mate\rt]\!\!, \dots,\!\! \lt[\matb  x_{\symt+1} \\  x_{\symt} \\  u_{\symt} \mate\rt]\!\rt]^{\!\!\!\top}$, we obtain
% \begin{align}\label{eq_proof_corr}
%     E_y C_\symt = [A,B] E_z C_\symt.
% \end{align}
% % which proves the lemma.
% \end{proof}

\subsection{A simple alternating joint estimation algorithm}\label{subsec_alt_method}

% Using the least-deviations loss function for its robustness properties, the combined regression of state and parameters on the observations gives
The combined estimation of the state, input and parameters, based on least absolute deviations (i.e. the $\ell_1$ norm loss function) is given by \cite{least_deviations_book}
\begin{align}\label{eq_regression_batch}
    &\hat A_\symt, \hat B_\symt, \hat Z_\symt = \\ \nonumber &\quad
    \argmin_{A, B, Z} \|[A, B] Z - \underline Y_\symt\|_{F1} + \|Z - \underline Z_\symt\|_{F1}.
\end{align}
In addition to the robustness provided by least absolute deviations, one may want to include a regularization for $\Theta = [A, B]$. In this case, with $\hat \Theta_\symt = [\hat A_\symt, \hat B_\symt]$, \eqref{eq_regression_batch} becomes
\begin{align}\label{eq_regression_batch_reg}
    \hat \Theta_\symt, \hat Z_\symt =
    \argmin_{\Theta, Z}&\; \|\Theta Z - \underline Y_\symt\|_{F1}  + \|Z - \underline Z_\symt\|_{F1} 
    \\ \nonumber
    &\; + {\|\Psi \vect(\Theta) \!-\! \mu\|}_{1}.
\end{align}
where $\Psi$ and $\mu$ are a matrix and a vector that can be chosen to tune the regularization term. We note that the problem \eqref{eq_regression_batch_reg} is %a generalization of
\begin{enumerate}[label=(\roman*)]
    \item equivalent to \eqref{eq_regression_batch} when $\Psi = \bb 0_{n(n+n_u)}^\top$, and $\mu = 0$,
    \item an $\ell_\infty$ robust and $\ell_\infty$ distributionally robust formulation of \eqref{eq_regression_batch} according to \cite[Theorem 3]{bertsimas2018_robust_opt} and \cite[Equation (4)]{multi_dro_chen2020}, when $\Psi = \epsilon I_{n(n+n_u)}$, and $\mu = \bb 0_{n(n+n_u)}$.
\end{enumerate}
Moreover, \eqref{eq_regression_batch_reg} is a maximum \emph{a posteriori} estimation problem with the prior belief that $\Psi \vect(\Theta) \approx \mu$ \cite{cdc_paper_dgs}.

The bilinear term $\Theta Z$ makes the optimization hard to solve. A common approach to circumvent this issue is to use a block coordinate descent method, which consists of an iterative optimization procedure that alternates between optimizing the estimate of $Z$ (for $\Theta$ fixed) and optimizing the estimate of the parameters $\Theta$ (for $Z$ fixed). Hence, $\Theta Z$ is linear in each subproblem. More precisely, the iteration $k$ of the optimization subproblems using the data at time step $t$ are given by

\begin{subequations}\label{eq_regression_batch_dro_final}
\begin{align}
    \label{eq_regression_batch_dro_state_final}
    \!\!\!\!\hat Z_{\symt\symk} \!&= 
    \argmin_{Z} \!\lt\|\!\lt[\matb \hat \Theta_{\symt\symk} \\  I_{n+n_u} \mate \rt]\! Z -\lt[\matb \underline Y_\symt \\  \underline Z_\symt \mate\rt] \!\rt\|_{F1}\!,\!\!\!\\
    \label{eq_regression_batch_dro_param_final}
    \!\!\!\!\hat \Theta_{\symt,\symk+1} \!&= 
    \argmin_{\Theta} \!{\|\Theta \hat Z_{\symt\symk} \!-\! \underline Y_\symt\|}_{F1} \!+\! {\|\Psi \vect(\Theta) \!-\! \mu\|}_{1},\!\!\! 
\end{align}
\end{subequations}

%\subsection{The state update}\label{section_state_upd}

% We will first focus on the optimization problem \eqref{eq_regression_batch_dro_state_final}, which aims to obtain estimates $\hat Z_{\symt\symk}$ based on observations $\underline Y_{\symt}$ and $\underline Z_{\symt}$, and current parameter estimates $\hat \Theta_{\symt\symk}$.

The update \eqref{eq_regression_batch_dro_state_final} estimates $\hat Z_{\symt\symk}$, the $Z$ portion of the correlation matrix at the iteration $k$ of the optimization using the data at time $t$. It does so based on observations $\underline Y_{\symt}$ and $\underline Z_{\symt}$, and the current parameter estimates $\hat \Theta_{\symt\symk}$. The following Lemma shows how \eqref{eq_regression_batch_dro_state_final} can be decomposed into simpler problems.

\begin{lemma}\label{lem_state_upd_sep}
The optimization problem \eqref{eq_regression_batch_dro_state_final} can be split into $n+n_u$ independent optimization problems, each only depending on one column $Z_{\symi}$ of $Z$.

\begin{align}\label{eq_regression_batch_dro_state_split}
    \argmin_{Z_{\symi}} \lt\|\!\lt[\matb \hat \Theta_{\symt\symk} \\ \  I_{n+n_u} \mate \rt]\!\! Z_{\symi} \!-\! \lt[\matb  \underline Y_{\symi, \symt} \\  \underline Z_{\symi, \symt} \mate\rt] \!\rt\|_1\!\!,\!
\end{align}
\end{lemma}
\begin{proof}
From the definition of the $\ell_1$ Frobenius norm, the cost in \eqref{eq_regression_batch_dro_state_final} is composed of $n+n_u$ terms, each depending only on one column of $\lt[\matb \hat \Theta_{\symt\symk} \\ \  I_{n+n_u} \mate \rt]\!\! Z$. This column is equal to $\lt[\matb \hat \Theta_{\symt\symk} \\ \  I_{n+n_u} \mate \rt]\!\! Z_{\symi}$. The proof is concluded by using the distribution property of the $\argmin$, i.e. $\argmin_{x,y} af(x) + bg(y) = \argmin_{x} f(x), \argmin_{y} g(y)$.
\end{proof}

Using \cref{lem_state_upd_sep}, we can update all columns $\hat Z_{\symi,\symt\symk}$ independently using iterative reweighting \cite{irls_convergence}. Given an integer $L_Z$, for an outer iteration $\symk$, iterative reweighting introduces the following $L_Z$ inner iterations indexed by $\syml = 1,\dots,L_Z$ to approximate the $\ell_1$ norm:
\begin{align}\label{eq_regression_batch_dro_state_ir}
    \!\!\hat Z_{\symi,\symt\symk, \syml+1} \!&=\! 
    \argmin_{Z_{\symi}} \lt\|\!\lt[\matb \hat \Theta_{\symt\symk} \\ \  I_{n+n_u} \mate \rt]\!\! Z_{\symi} \!-\! \lt[\matb  \underline Y_{\symi, \symt} \\  \underline Z_{\symi, \symt} \mate\rt] \!\rt\|_{W_{\symi \symt\symk\syml}}^2\!\!, \!\!
    \\ \nonumber \!\!\! 
    \!\!W_{\symi\symt\symk\syml}^{-1} \!&=\! \sqrt{\mt{diag} \!\lt(\lt[\matb
     \hat \Theta_{\symt\symk} \hat Z_{\symi,\symt\symk\syml} \!-\! \underline Y_{\symi, \symt} \\
     \hat Z_{\symi,\symt\symk\syml} \!-\!  \underline Z_{\symi, \symt}
    \mate\rt]\rt)^{\!\!2} + \alpha I_{m}},
\end{align}
where $0 < \alpha \ll 1$ is a small parameter introduced for numerical stability. After the last inner iteration, $\hat Z_{\symt\symk, L_{Z}}$ is used as an approximate solution to \eqref{eq_regression_batch_dro_state_final}. Note that for each time step $\symt$, we now have a double nested loop over $\symk$ and $\syml$, which may be very slow in practice. A remedy for this issue is described in \cref{section_conv}.

%\subsection{The parameter update}\label{section_param_upd}

Next, we analyze the parameter update \eqref{eq_regression_batch_dro_param_final}, which estimates $\hat \Theta_{\symt,\symk+1}$ based on $\hat Z_{\symt\symk, L_{Z}}$. Unlike \eqref{eq_regression_batch_dro_state_final}, the problem  \eqref{eq_regression_batch_dro_param_final} cannot easily be split into sub problems. However, one can vectorize the parameters to simplify \eqref{eq_regression_batch_dro_param_final} using iterative reweighting. %Instead, vectorized variables are needed to provide a closed-form solution. We therefore define the following variables

\begin{align}\label{eq_vect_defs}
    \theta &= \vect(\Theta), \; \hat \theta_{\symt,\symk+1} = \vect(\hat \Theta_{\symt,\symk+1}), \;
    \hat{\bb Z}_{\symt\symk} = \hat Z_{\symt\symk}^\top \otimes I_n
\end{align}
Since $\vect(XY) = (Y^\top \otimes I)\vect(Y)$, \eqref{eq_vect_defs} gives

\begin{align}\label{eq_vect_prop}
   \vect(\Theta \hat Z_{\symt\symk} - \underline Y_\symt) = \hat{\bb Z}_{\symt\symk} \theta - \vect(\underline Y_\symt),
\end{align}
and, therefore, \eqref{eq_regression_batch_dro_param_final} can be written as
\begin{align}\label{eq_vect_opt_stack}
    \argmin_\theta \lt\|\lt[\matb \hat{\bb Z}_{\symt\symk} \\ \Psi \mate \rt] \theta - \lt[\matb\vect(\underline Y_\symt) \\ \mu \mate\rt]\rt\|_1.
\end{align}

Similar to \eqref{eq_regression_batch_dro_state_ir}, problem \eqref{eq_vect_opt_stack} can be solved for each iteration $\symk$ using iteratively reweighted inner iterations $\syml = 1,\dots,L_\Theta$ to approximate the $\ell_1$ norm \cite{irls_convergence}, i.e.

\begin{align}\label{eq_regression_batch_dro_param_ir}
    \hat \theta_{\symt,\symk+1, \syml+1} &= \argmin_\theta \lt\|\lt[\matb \hat{\bb Z}_{\symt\symk} \\ \Psi \mate \rt] \theta - \lt[\matb\vect(\underline Y_\symt) \\ \mu \mate\rt]\rt\|_{V_{\symt,\symk+1,\syml}}^2\!\!, \!\!\!\!\!\!\!\!
    \\ \nonumber \!\!\! 
    V_{\symt,\symk+1,\syml}^{-1} \!&=\! \sqrt{\mt{diag} \!\lt(\lt[\matb
    \hat{\bb Z}_{\symt\symk} \hat \theta_{\symt,\symk+1,\syml} \!-\! \vect(\underline Y_\symt) \\
    \Psi \hat \theta_{\symt,\symk+1,\syml} \!-\! \mu
    \mate\rt]\rt)^{\!\!2} + \alpha I_{nm+M}},
\end{align}

The quantity $\hat \theta_{\symt\symk, L_{\Theta}}$ obtained in the last iteration is an approximate solution to \eqref{eq_regression_batch_dro_param_final}.

\subsection{The overall algorithm}\label{subsec_algo1}

The estimation procedure alternates between \eqref{eq_regression_batch_dro_state_final} and \eqref{eq_regression_batch_dro_param_final}, solved iteratively using \eqref{eq_regression_batch_dro_state_ir} and \eqref{eq_regression_batch_dro_param_ir}, respectively. The full implementation of both outer and inner loops at each time instant is provided in \cref{algo_airls}.
\begin{algorithm}[H]
\caption{}\label{algo_airls}
\begin{algorithmic}
\State $\underline C_0 = I_m$
\For{$\symt = 1,\dots,N$}
\State $\underline C_{\symt} \gets \beta \underline C_{\symt-1} + [\tilde x_{\symt}^\top, \tilde x_{\symt-1}^\top, \tilde u_{\symt-1}^\top]^\top [\tilde x_{\symt}^\top, \tilde x_{\symt-1}^\top, \tilde u_{\symt-1}^\top]$
\State $\hat \Theta_{\symt,\symk = 0} \gets \hat \Theta_{\symt-1,K}$
\State $\hat Z_{\symt,\symk = 0} \gets \hat Z_{\symt-1, K}$
\For{$\symk = 1,\dots,K$}
\State $\hat Z_{\symt\symk,\syml = 0} \gets \hat Z_{\symt,\symk-1}$
\State $\hat \Theta_{\symt\symk,\syml = 0} \gets \hat \Theta_{\symt,\symk-1}$
\For{$\syml = 1,\dots,L_Z$}
\State {\textbf{minimize} \eqref{eq_regression_batch_dro_state_ir} to obtain $\hat Z_{\symt\symk\syml}$}
\EndFor
\State $\hat Z_{\symt\symk} \gets \hat Z_{\symt\symk, L_Z}$
\For{$\syml = 1,\dots,L_\Theta$}
\State {\textbf{minimize} \eqref{eq_regression_batch_dro_param_ir} to obtain $\hat \Theta_{\symt\symk\syml }$}
\EndFor
\State $\hat \Theta_{\symt\symk} \gets \hat \Theta_{\symt\symk, L_\Theta}$
\EndFor
\State $\underline C_\symt \gets E_y^\top E_y \underline C_\symt + E_z^\top \hat Z_{\symt,K}$
\EndFor
\end{algorithmic}
\end{algorithm}

\cref{algo_airls} provides a robust solution for the joint state/input and parameter estimation problem using a fixed-size matrices, which suits online application. However, the nested loops are often too slow for real-time application. We will therefore not study the convergence of \cref{algo_airls}. Instead, we will study the convergence of a more computationally efficient version, presented in the next section.

\section{Alternating and Iteratively-reweighted Recursive Least Squares (AIRLS)}\label{section_conv}

In this section, we will first show how to easily compute the optimizers of \eqref{eq_regression_batch_dro_state_ir} and \eqref{eq_regression_batch_dro_param_ir}, and then prove that \cref{algo_airls} converges when $K = L_Z = L_\Theta = 1$.
Problems \eqref{eq_regression_batch_dro_state_ir} and \eqref{eq_regression_batch_dro_param_ir}, admit a closed-form solution, as discussed in the following.
\begin{definition}
For any pair of matrices $X$ and $W$ such that $WX$ exists and has full column rank, the weighted pseudo-inverse is $X^\dagger_W = (X^\top W X)^{-1} X^\top W$.
\end{definition}
\noindent
Note that, by construction $X X^\dagger_W X = X$.

The problem \eqref{eq_regression_batch_dro_param_ir} is quadratic and solved by
\begin{align}\label{eq_regression_batch_dro_param_ir_sol}
    \hat \theta_{\symt,\symk+1, \syml+1} \!&=\! \lt[\matb \hat{\bb Z}_{\symt\symk}\! \\ \Psi \mate \rt]^{\dagger}_{V_{\symt,\symk+1,\syml}} \lt[\matb\vect(\underline Y_\symt) \\ \mu \mate\rt]\!.
\end{align}
Computing \eqref{eq_regression_batch_dro_param_ir_sol} amounts to solve a linear system with as many equations as the number of parameters in $\Theta$. %It can challenging to compute for large systems but it can still be performed in real-time for thousands of parameters.
The problem \eqref{eq_regression_batch_dro_state_ir} is composed of $N$ multivariate optimization problems, which can all be solved by an oblique projection of the $\symi^{th}$ column of $\underline C_{\symt}$ on the null space of $[-I_n, \hat \Theta_{\symt\symk}]$, weighed by $W_{\symi\symt\symk\syml}$.

\begin{theorem}\label{thm_proj}
The problem \eqref{eq_regression_batch_dro_state_ir} is solved by
\begin{align}\label{eq_rrairls_sol_z}
    \hat Z_{\symi,\symt\symk, \syml+1} &= E_z P_{\symi\symt\symk\syml}\underline C_{\symi, \symt},
\end{align}
where
\begin{align}\label{eq_proof_thm1_proj_def}
    P_{\symi\symt\symk\syml} \!=\! I_m - (([-I_n, \hat \Theta_{\symt\symk}]^\top)^\dagger_{W_{\symi\symt\symk\syml}})^\top [-I_n, \hat \Theta_{\symt\symk}].
\end{align}
\end{theorem}
\begin{proof}
For a basis $\bb B$ such that $\mt{range}(\bb B) = \mt{null}([-I_n, \hat \Theta_{\symt\symk}])$, $P_{\symi\symt\symk\syml}$ provides the weighted least squares solution \cite{meyer2000matrix}
\begin{align}\label{eq_proof_thm1_opt_proj}
    P_{\symi\symt\symk\syml} [\underline C_\symt]_\symi = \argmin_{d}\; & \lt\|d - [\underline C_\symt]_\symi \rt\|_{W_{\symi\symt\symk\syml}}^2, \\  \nonumber
    &\quad \mt{s.t. } d \in \mt{range}(\bb B).
\end{align}
Choosing $d = [Y^\top, Z^\top]^\top$, the problem \eqref{eq_proof_thm1_opt_proj} becomes
\begin{align}\label{eq_proof_thm1_opt}
    \!P_{\symi\symt\symk\syml} [\underline Y_{\symi, \symt}^\top, \underline Z_{\symi, \symt}^\top]^{\!\top} \!\!&=\! \argmin_{Y_\symi, Z_\symi} \!\lt\|\![Y_\symi^\top, Z_\symi^\top]^{\!\top} \!\! -\! [\underline Y_{\symi, \symt}^\top, \underline Z_{\symi, \symt}^\top]^{\!\top}\! \rt\|_{W_{\symi\symt\symk\syml}}^2\!\!\!,\!\! \nonumber \\
    &\quad\; \mt{s.t. } [-I_n, \hat \Theta_{\symt\symk}][Y_\symi^\top\!\!, Z_\symi^\top]^{\!\top} \!\!= \bb 0_n.
\end{align}
Plugging the constraint to replace $Y_{\symi}$ in \eqref{eq_proof_thm1_opt} yields exactly \eqref{eq_regression_batch_dro_state_ir} (with $\symt+1$ instead of $\symt$).
\end{proof}
\begin{corollary}\label{cor_proj}
The estimate of the state and input at a particular time step $t$ is given by
\begin{align}
    [\hat x_{\symt+1, \symk, \syml+1}^\top, \hat x_{\symt\symk, \syml+1}^\top, \hat u_{\symt\symk, \syml+1}^\top]^\top = P_{x,\symt\symk\syml} [\tilde x_{\symt+1}^\top, \tilde x_{\symt}^\top, \tilde u_{\symt}^\top]^\top,
\end{align}
where $P_{x,\symt\symk\syml}$ is defined by \eqref{eq_proof_thm1_proj_def} with
\begin{align}
    W_{x,\symt\symk\syml}^{-1} \!&=\! \sqrt{\mt{diag} \!\lt(\lt[\matb
     \hat \Theta_{\symt\symk} [\hat x_{\symt\symk\syml}^\top, \hat u_{\symt\symk\syml}^\top]^\top \!\!- \tilde x_{\symt+1} \\
     [\hat x_{\symt\symk\syml}^\top, \hat u_{\symt\symk\syml}^\top]^\top \!\!-  [\tilde x_{\symt}^\top, \tilde u_{\symt}^\top]^\top
    \mate\rt]\rt)^{\!\!2} \!\!+\! \alpha I_{m}},
\end{align}
\end{corollary}
\begin{proof}
The proof is given by replacing $[\underline C_{\symt}]_\symi$ by $[ \tilde x_{\symt+1}^\top,  \tilde x_{\symt}^\top,  \tilde u_{\symt}^\top]^\top$ and $\hat Z_{\symi, \symt\symk\syml}$ by $[\hat x_{\symt\symk\syml}^\top,  \hat u_{\symt\symk\syml}^\top]^\top$ in \eqref{eq_regression_batch_dro_state_ir} and in the proof of \cref{thm_proj}.
\end{proof}

Computing the projector $P_{\symi\symt\symk\syml}$ may be expensive due to the pseudo-inverse \eqref{eq_proof_thm1_proj_def}. However, \eqref{eq_rrairls_sol_z} only requires the projection of the correlation matrix $\underline C_\symt$ (i.e. $n_u + 2n$ vectors), which is much faster to compute. 

% In order to benefit from previous estimates of $C_\symt$ and not restart from zero at each time-step, one can use
% \begin{align}\label{eq_rrairls_sol_c}
%     &\hat C_{\symi,\symt\symk, \syml+1}
%     \\ \nonumber
%     &\;= P_{\symi\symt\symk\syml}[\beta \hat C_{\symi,\symt-1,KL_\Theta} + [\tilde x_{\symt+1}^\top, \tilde x_{\symt}^\top, \tilde u_{\symt}^\top]^\top [\tilde x_{\symt+1}^\top, \tilde x_{\symt}^\top, \tilde u_{\symt}^\top]]_\symi,
% \end{align}
% and update $\hat Z$ with $\hat Z_{\symi,\symt\symk, \syml+1} = E_z \hat C_{\symi,\symt\symk, \syml+1}$. This is an approximation of \eqref{eq_rrairls_sol_z}, but numerical evidence shows that it leads to more accurate results.

\subsection{AIRLS estimator}

The AIRLS algorithm is defined as \cref{algo_airls} with $K = L_Z = L_\Theta = 1$ and where \eqref{eq_rrairls_sol_z} and \eqref{eq_regression_batch_dro_param_ir_sol} are used for computing the optimizers of \eqref{eq_regression_batch_dro_state_ir} and \eqref{eq_regression_batch_dro_param_ir}, respectively. It has the computational advantage of replacing the nested loops in \cref{algo_airls} with one-step updates. In a sense, the robustness provided by the $\ell_1$ cost and the regularization term in \eqref{eq_regression_batch_reg} help compensate for the unfinished loops. In the sequel, for simplicity, we will drop the subscripts $\symk = 1$ and $\syml = 1$.% in $\hat \Theta_{\symt\symk\syml}$, $\hat Z_{\symt\symk\syml}$, $\underline C_{\symt\symk\syml}$, and $\hat P_{\symi\symt\symk\syml}$.

%In Algorithm \ref{algo_airls}, for the variables $\hat \Theta$ and $\hat Z$, the first subscript refers to the time step, the second subscript refers to the alternating iteration $k$, and the third subscript, if present, refers to the iteratively reweighted sub-iteration. Algorithm \ref{algo_airls} uses only one alternating iteration at each time step, and only one iteratively reweighted sub-iteration for each alternating iteration.

% \subsection{convergence}
\begin{definition}
Let $\underline{\bar \Gamma}_{\symt}$ represent an average measurement of the system such that the corresponding asymptotic correlation matrix $\sum_{\symi = 0}^\infty \beta^i \underline{\bar \Gamma}_{\symi}$ is equal to $\underline{C}_{\symt}$.
\end{definition}

Similarly to the inverse power method \cite{rhode2014_rgtls}, AIRLS needs $\underline{C}_{\symt}$ to be full rank to ensure convergence. Numerical experiments in \cref{section_results} show that it also converges when $\underline{C}_{\symt}$ has rank $n+n_u$, i.e. in the noiseless case.

\begin{theorem}[convergence]\label{thm_conv}
With bounded measurements $\tilde \Gamma_\symt \preceq \gamma_{\mt{max}} I_m$ for all $\symt$, and with $\underline{\bar \Gamma}_{\symt} \succeq \gamma_{\mt{min}} I_m$, a forgetting factor satisfying $1 - \beta \leq \gamma_{\mt{max}}^{-2} \gamma_{\mt{min}}^{2}$ guarantees that the AIRLS update converges and can only decrease unweighted residuals of \eqref{eq_regression_batch_dro_param_ir}, i.e.
\begin{align}\label{eq_proof_res_def}
    \!\!\!R(\hat \Theta_\symt, \underline C_\symt) \!=\! \lt[\matb \hat{\bb Z}_\symt \\ \Psi\mate\rt] \!\! \vect(\hat \Theta_\symt) \!-\! \lt[\matb \underline Y_\symt \\ \mu \mate \rt] \!=\! \lt[\matb \vect([\unaryminus I_n, \hat \Theta_\symt] \underline C_\symt) \\ \Psi \vect(\hat \Theta_\symt) - \mu \mate \rt]\!\!.\!
\end{align}

\end{theorem}

\begin{proof}
Using \eqref{eq_vect_prop} and \eqref{eq_regression_batch_dro_param_ir_sol} to write $R(\hat \Theta_{\symt+1}, \underline C_\symt)$ depending on $\underline C_\symt$ and $\hat \Theta_{\symt}$ yields
\begin{align}\label{eq_proof_regression_res_param}
    &\lt[\matb \hat{\bb Z}_\symt \\ \Psi\mate\rt] \!\! \vect(\hat \Theta_{\symt+1}) \!-\! \lt[\matb \underline Y_\symt \\ \mu \mate \rt] = 
    \\
    \nonumber & \quad\quad\quad
    \!\lt(\!I_{n(n+n_u)+M} \!-\!  \lt[\matb \hat{\bb Z}_{\symt}\! \\ \Psi \mate \rt] \!\! \lt[\matb \hat{\bb Z}_{\symt}\! \\ \Psi \mate \rt]^{\!\dagger}_{\!V_\symt} \!\rt)\!\!  \lt[\matb \hat{\bb Z}_\symt \\ \Psi\mate\rt] \!\! \vect(\hat \Theta_{\symt}) \!-\! \lt[\matb \underline Y_\symt \\ \mu \mate \rt], \nonumber
\end{align}
By definition, $X X^\dagger_W$ is an oblique projection for any $X$ and $W$, which means that $\|R(\hat \Theta_{\symt+1}, \underline C_\symt)\|_2^2 \leq \|R(\hat \Theta_\symt, \underline C_\symt)\|_2^2$.

We now compare $\|R(\hat \Theta_\symt, \underline C_\symt)\|_2^2$ to $\|R(\hat \Theta_\symt, \underline C_{\symt+1})\|_2^2$. \newline First, note that the second block in \eqref{eq_proof_res_def} does not contain $\underline C_\symt$. This means that $\|R(\hat \Theta_\symt, \underline C_\symt)\|_2^2 - \|R(\hat \Theta_\symt, \underline C_{\symt+1})\|_2^2$ is equal to $\|\vect([-I_n, \hat \Theta_\symt] \underline C_\symt)\|_2^2 - \|\vect([-I_n, \hat \Theta_\symt] \underline C_{\symt+1})\|_2^2$.\newline Second, we write the following decomposition:
\begin{align}\label{eq_proof_res_split}
    \|\vect([-I_n, \hat \Theta_\symt] \underline C_{\symt+1})\|_2^2 = \sum_{\symi = 1}^m \|[-I_n, \hat \Theta_\symt] \underline C_{\symi,\symt+1}\|_2^2.
\end{align}
Third, we write the closed form solution of \eqref{eq_regression_batch_dro_state_ir} (which is equal to \eqref{eq_rrairls_sol_z}) as
\begin{align}\label{eq_regression_batch_dro_param_ir_sol_pseudoinv}
    \hat Z_{\symi,\symt+1} \!&=\! \lt[\matb \hat \Theta_\symt \\ I_{n+n_u} \mate \rt]^{\dagger}_{W_{\symi,\symt}} \!\! (\beta \underline C_{\symi,\symt} + \tilde \Gamma_{\symi,\symt+1}) .
\end{align}
Similar to \eqref{eq_proof_regression_res_param}, we can construct a projection with $X = [\hat \Theta_t^\top, I_{n+n_u}]^\top$ instead of $[\hat{\bb Z}_{\symt}^\top, \Psi^\top]^\top$:
\begin{align}\label{eq_proof_res_proj_c}
    %[-I_n, \hat \Theta_\symt] \underline C_{\symi,\symt+1} = [-I_n, \hat \Theta_\symt] () (\beta \underline C_{\symi,\symt} + \tilde \Gamma_{\symt+1})
    &X \hat Z_{\symi,\symt+1} - (\beta \underline C_{\symi,\symt} + \tilde \Gamma_{\symi,\symt+1}) = 
    \\ \nonumber
    &\quad\quad \lt(I_m -  X X^\dagger_{W_{\symi\symt}} \rt) \lt( X E_z - I_m \rt) (\beta \underline C_{\symi,\symt} + \tilde \Gamma_{\symi,\symt+1}).
\end{align}
Moreover, because of the last step of \cref{algo_airls},
%\begingroup\setlength{\belowdisplayskip}{-3pt}
\begin{align}\nonumber
\|[-I_n, \hat \Theta_\symt] \underline C_{\symi,\symt+1}\|_2^2 &= \lt\|\hat \Theta_\symt \hat Z_{\symi,\symt+1} \!-\! E_y (\beta \underline C_{\symi,\symt} \!+\! \tilde \Gamma_{\symi,\symt+1})\rt\|_2^2 
\\ \nonumber
&\leq \lt\|\!\lt[ \matb \hat \Theta_\symt \hat Z_{\symi,\symt+1} \!-\! E_y (\beta \underline C_{\symi,\symt} \!+\! \tilde \Gamma_{\symi,\symt+1}) \\
\hat Z_{\symi,\symt+1} \!-\! E_z (\beta \underline C_{\symi,\symt} \!+\! \tilde \Gamma_{\symi,\symt+1}) \mate \rt]\!\rt\|_2^2
\\
&\leq \lt\|\! \lt[\matb \hat \Theta_\symt \\ I_{n+n_u} \mate \rt] \hat Z_{\symi,\symt+1} \!-\! (\beta \underline C_{\symi,\symt} \!+\! \tilde \Gamma_{\symi,\symt+1}) \rt\|_2^2\!\!.
\label{eq_proof_res_ineq_1}
\end{align}
%\endgroup
Combining \eqref{eq_proof_res_ineq_1} with the projection \eqref{eq_proof_res_proj_c} yields
\begin{align}\label{eq_proof_res_ineq_2}
\|[-I_n, \hat \Theta_\symt] \underline C_{\symi,\symt+1}\|_2^2 &\leq \|[-I_n, \hat \Theta_\symt] (\beta \underline C_{\symi,\symt} + \Gamma_{\symi,\symt+1})\|_2^2 
\end{align}
because $( X E_z - I_m ) = \lt[ \matb [-I_n, \hat \Theta_\symt] \\ \bb 0_{(n+n_u) \times m} \mate \rt]$. Hence, $\|R(\hat \Theta_\symt, \underline C_{\symt+1})\|_2^2 \leq \|R(\hat \Theta_\symt, \beta \underline C_\symt + \tilde \Gamma_{\symt+1})\|_2^2$.

Finally, to prove that $\|R(\hat \Theta_{\symt+1}, \underline C_{\symt+1})\|_2^2 \leq \|R(\hat \Theta_{\symt}, \underline C_\symt)\|_2^2$ we need $\|[-I_n, \hat \Theta_\symt] (\beta \underline C_{\symt} + \Gamma_{\symt+1})\|_F^2 \leq \|[-I_n, \hat \Theta_\symt] \underline C_{\symt}\|_F^2 $, which is true if 
\begin{align}\label{eq_proof_res_ineq_3}
\|[-I_n, \hat \Theta_\symt] \Gamma_{\symt+1}\|_F^2 \leq (1-\beta^2)\|[-I_n, \hat \Theta_\symt] \underline C_{\symt}\|_F^2.
\end{align}
The assumption that $\tilde \Gamma_\symt \preceq \gamma_{\mt{max}} I_m$ and that $\underline{\bar \Gamma}_{\symt} \succeq \gamma_{\mt{min}} I_m$ ensure \eqref{eq_proof_res_ineq_3} if
\begingroup\setlength{\abovedisplayskip}{-1pt}
\begin{align}\label{eq_proof_beta_bound}
    \frac{\gamma_{\mt{max}}^2}{\gamma_{\mt{min}}^2} \leq \frac{1 - \beta^2}{(1-\beta)^2},
\end{align}
\endgroup
which is guaranteed if $1 - \beta \leq \gamma_{\mt{max}}^{-2} \gamma_{\mt{min}}^{2}$. The function $\|R(\hat \Theta_{\symt}, \underline C_\symt)\|_2^2$ is therefore decreasing and lower bounded, which proves the theorem
\end{proof}

\section{Numerical experiments}\label{section_results}

In this section, we will compare the parameter estimates and state predictions in the asymptotic regime using AIRLS with $\Psi = 10^{-3} I_8$ and $\mu = \bb 0_8$, the EKF from \cite{ekf_kullberg2021online}, the RTLS from \cite{rhode2014_rgtls} and subspace identification. For the latter, we use the batch method provided by the function \emph{n4sid} in MATLAB \cite{ref_n4sid}.\footnote{Batch estimation is expected to outperform any recursive implementation for the same sample size.}
We use the system
\begin{align}
    x_{\symt+1} = \lt[\matb 0.8 & -0.25 \\ -0.25 & 0.25 \mate\rt]x_{\symt} + \lt[\matb 10 & 2 \\ 2 & 10 \mate\rt] u_{\symt},
\end{align}
with a random persistent excitation $u \sim \mc N(\bb 0_2,0.01 I_2)$. We add weak Gaussian measurement noise with a signal to noise ratio of 100 to all samples, and much stronger noise (uniformly distributed in $[-0.2, 0.2]$) for a small portion of randomly chosen samples, varying between $0.02\%$ and $5\%$ of all samples. The points affected by the strong noise are outliers. For each proportion of outliers, we average the estimates of 10 different experiments.

\cref{fig:errs} shows the relative Frobenius error
\begin{align}
    \epsilon_F = \frac{\|[A,B] - [\hat A_N, \hat B_N]\|_F}{\|[A,B]\|_F}, \; N=50000,
\end{align}
for all 4 methods and various proportions of outliers. We observe that both subspace identification and EKF get large errors with as low as 0.1\% outliers. This means that they may perform poorly even with the help of an outlier detection system that is not 100\% accurate. The RTLS is much more robust, but still performs much worse than AIRLS. \cref{fig:preds} shows the state estimation for $\symt > 50000$, i.e. when the parameters have converged. The error on parameters manifests as excessive smoothing of the state estimation.

We conclude by highlighting that with standard Gaussian noise and no outliers, all methods achieve similar performance, and that without any noise, all methods have 100\% accuracy.

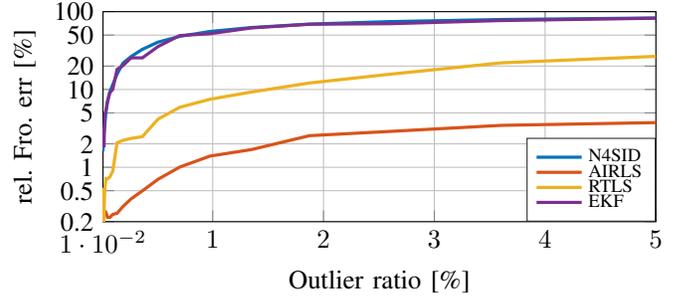
\begin{figure}[t!]
    \centering
\definecolor{mycolor1}{rgb}{0.00000,0.44700,0.74100}%
\definecolor{mycolor2}{rgb}{0.85000,0.32500,0.09800}%
\definecolor{mycolor3}{rgb}{0.92900,0.69400,0.12500}%
\definecolor{mycolor4}{rgb}{0.49400,0.18400,0.55600}%
\begin{tikzpicture}

\begin{axis}[%
width=0.83\columnwidth,
height=1.1in,
scale only axis,
xmin=0.01,
xmax=5,
xtick={  0.01, 1,   2,   3,  4,  5},
xlabel={Outlier ratio [\%]},
ymode=log,
ymin=0.2,
ymax=100,
ylabel={rel. Fro. err [\%]},
ytick={  0.2, 0.5, 1,   2,   5,  10,  20,  50, 100},
yticklabels={  0.2, 0.5, 1,   2,   5,  10,  20,  50, 100},
yminorticks=true,
axis background/.style={fill=white},
xmajorgrids,
ymajorgrids,
yminorgrids,
legend style={at={(0.766,0)}, anchor=south west, legend cell align=left, align=left, draw=white!15!black, row sep=-0.16cm, font=\scriptsize}
]
\addplot [color=mycolor1, line width=1.2pt]
  table[row sep=crcr]{%
0.01	1.61910051435368\\
0.0138691885653003	2.03359101645307\\
0.0192354391459856	2.80972253691918\\
0.0266779932652033	3.52929626621682\\
0.0370002119138915	5.14651220171146\\
0.0513162915989831	7.07300623334697\\
0.0711715324658232	9.2344696354313\\
0.0987091404249892	11.5766761587765\\
0.136901568167288	15.721174746537\\
0.189871366379743	21.7339429160201\\
0.263336178347187	26.4849203779372\\
0.365225911356268	32.7716586410863\\
0.506538703353372	40.6005929714368\\
0.702528079243062	48.146853777119\\
0.974349440344024	55.6387898925703\\
1.35134361166261	62.5672122521583\\
1.87420393666626	69.2211505514845\\
2.59936878074525	74.5267141660979\\
3.60511357709105	79.2663635040718\\
5	82.9600474415179\\
};
\addlegendentry{N4SID}

\addplot [color=mycolor2, line width=1.2pt]
  table[row sep=crcr]{%
0.01	0.291410112203469\\
0.0138691885653003	0.280375670172788\\
0.0192354391459856	0.266955715636834\\
0.0266779932652033	0.273438727777797\\
0.0370002119138915	0.263220402652873\\
0.0513162915989831	0.22732793680007\\
0.0711715324658232	0.226061504781267\\
0.0987091404249892	0.248228000136179\\
0.136901568167288	0.256001099248256\\
0.189871366379743	0.312388941735475\\
0.263336178347187	0.392219983426649\\
0.365225911356268	0.497251792607906\\
0.506538703353372	0.700422214791405\\
0.702528079243062	1.0054673677974\\
0.974349440344024	1.3910220190101\\
1.35134361166261	1.69061063175565\\
1.87420393666626	2.54809311237077\\
2.59936878074525	2.88972048397072\\
3.60511357709105	3.45756267597291\\
5	3.75213714133917\\
};
\addlegendentry{AIRLS}

\addplot [color=mycolor3, line width=1.2pt]
  table[row sep=crcr]{%
0.01	0.545828268074622\\
0.0138691885653003	0.166137701319153\\
0.0192354391459856	0.253407575221162\\
0.0266779932652033	0.52554244924441\\
0.0370002119138915	0.715666343814436\\
0.0513162915989831	0.71017359850859\\
0.0711715324658232	0.748240580874207\\
0.0987091404249892	0.893349627527041\\
0.136901568167288	2.06411064378967\\
0.189871366379743	2.21884210684622\\
0.263336178347187	2.34800591297986\\
0.365225911356268	2.47486518570053\\
0.506538703353372	4.15161763837868\\
0.702528079243062	5.89953634581419\\
0.974349440344024	7.47793356198546\\
1.35134361166261	9.28056253935202\\
1.87420393666626	12.1097658642429\\
2.59936878074525	15.7020382559705\\
3.60511357709105	21.9812421655024\\
5	26.6514223468153\\
};
\addlegendentry{RTLS}

\addplot [color=mycolor4, line width=1.2pt]
  table[row sep=crcr]{%
0.01	3.20209628090576\\
0.0138691885653003	2.3930207973149\\
0.0192354391459856	1.81155476730571\\
0.0266779932652033	4.88764065516469\\
0.0370002119138915	5.37199190267368\\
0.0513162915989831	6.98128932024942\\
0.0711715324658232	9.02856031002161\\
0.0987091404249892	10.0536816596031\\
0.136901568167288	18.1082040278798\\
0.189871366379743	20.2399996210048\\
0.263336178347187	25.497830279423\\
0.365225911356268	25.4734871181505\\
0.506538703353372	35.5746378471082\\
0.702528079243062	49.2999404411758\\
0.974349440344024	52.0227380169181\\
1.35134361166261	61.833955235916\\
1.87420393666626	68.893413374871\\
2.59936878074525	70.0328179012107\\
3.60511357709105	76.880617737516\\
5	82.3608269849719\\
};
\addlegendentry{EKF}

\end{axis}
\end{tikzpicture}%
    \caption{Relative Frobenius error of the parameter estimates for various methods using data with a proportion of outliers up to 5\%. The vertical axis is in log scale.}
    \label{fig:errs}   
\end{figure}
\begin{figure}[t!]
    \centering
    \includegraphics[width=\columnwidth]{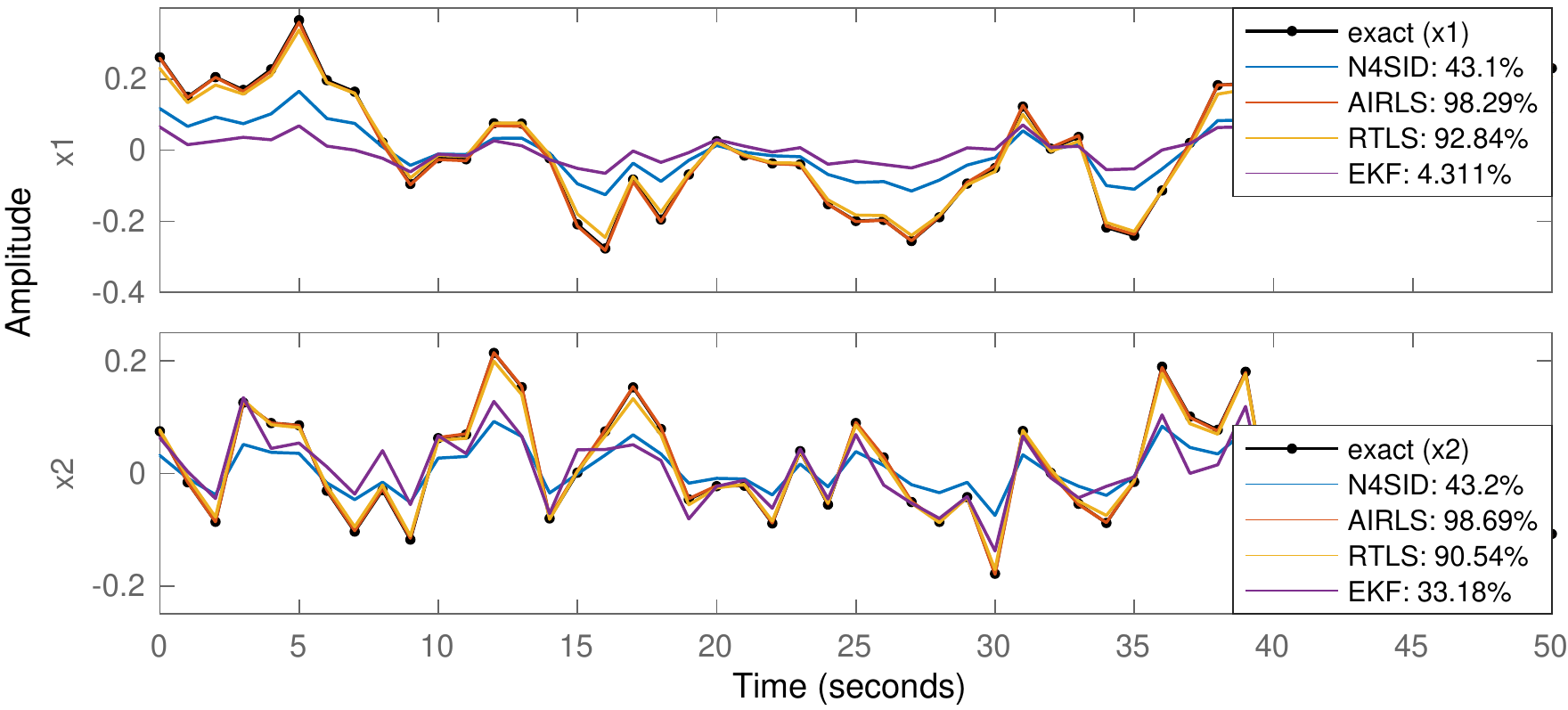}
    \caption{State estimation after $N$ steps using various methods with 1\% outliers. The plot of the exact state is overlapping with the one of the AIRLS estimate.}
    \label{fig:preds}
\end{figure}

\section{Conclusions}\label{section_conclu}

We show that AIRLS, an algorithm that combines recursive, alternating, and iteratively-reweighted least squares, converges and allows one to perform robust and online joint state/input and parameter estimation for linear systems. Numerical experiments show that the accuracy of the AIRLS estimates is higher than state-of-the-art methods in the presence of outliers.

Future work includes extending AIRLS to more general loss functions and noise distributions. Practical applications, including power systems and self driving cars will also be addressed.

\bibliographystyle{IEEEtran}
\bibliography{references}

\end{document}